\documentclass[11pt,oneside,english]{amsart}
\usepackage[T1]{fontenc}
\usepackage[latin9]{inputenc}
\usepackage{amsthm}
\usepackage{graphicx}
\usepackage{fancyhdr}

\makeatletter
\numberwithin{equation}{section}
\numberwithin{figure}{section}
\theoremstyle{plain}
\newtheorem{theorem}{Theorem}
  \theoremstyle{plain}

\newtheorem{corollary}[theorem]{Corollary}
\newtheorem{observation}[theorem]{Observation}

\newcommand{\gx}{G^{\times}}
\makeatletter

\newcommand{\Rmnum}[1]{\expandafter\@slowromancap\romannumeral #1@}
\makeatother
\newcommand{\C}{$\mathcal {C}$}
\newcommand{\R}{$\mathcal {R}$}

\usepackage{fullpage}

\makeatother

\usepackage{babel}

\begin{document}

\title{Structural properties of 1-planar graphs and an application to acyclic edge coloring}
\thanks{This is a cover-to-cover translation of its Chinese version that has been published in Scientia Sinica Mathematica.}
\thanks{Research supported by grants NNSF (10871119, 10971121) and RFDP (200804220001) of China}

\author{Xin Zhang}
\address{School of Mathematics\\ Shandong University\\ Jinan 250100, China.}
\email{sdu.zhang@yahoo.com.cn, xinzhang@mail.sdu.edu.cn}

\author{Guizhen Liu}
\address{School of Mathematics\\ Shandong University\\ Jinan 250100, China.}
\email{gzliu@sdu.edu.cn}

\author{Jian-Liang Wu}
\address{School of Mathematics\\ Shandong University\\ Jinan 250100, China.}
\email{jlwu@sdu.edu.cn}

\begin{abstract}
A graph is called $1$-planar if it can be drawn on the plane so that each edge is crossed by at most one other edge. In this paper, we establish a local property of $1$-planar graphs which describes the structure in the neighborhood of small vertices (i.e. vertices of degree no more than seven). Meanwhile, some new classes of light graphs in $1$-planar graphs with the bounded degree are found. Therefore, two open problems presented by Fabrici and Madaras [The structure of 1-planar graphs, Discrete Mathematics, 307, (2007), 854--865] are solved. Furthermore, we prove that each $1$-planar graph $G$ with maximum degree $\Delta(G)$ is acyclically edge $L$-choosable where
$L=\max\{2\Delta(G)-2,\Delta(G)+83\}$.

\vspace{1mm}\noindent\textit{Keywords}: $1$-planar graph; light
graph; acyclic edge coloring.
\end{abstract}

\maketitle

\noindent Please cite this published article as: \emph{X. Zhang, G. Liu, J.-L. Wu. Structural properties of 1-planar graphs and an
application to acyclic edge coloring. Scientia Sinica Mathematica, 2010, 40, 1025--1032}.

\section{Introduction}

In this paper, all graphs considered are finite, simple and undirected. We use $V(G)$, $E(G)$, $\delta(G)$ and $\Delta(G)$ to denote the vertex set, the edge set, the minimum degree and the maximum degree of a graph $G$, respectively. Denote $v(G)=|V(G)|$ and $e(G)=|E(G)|$. Let $d_G(v)$ (or $d(v)$ for simple) denote the degree of vertex $v\in V(G)$. A $k$-, $k^+$- and $k^-$-$vertex$ is a vertex of degree $k$, at least $k$ and at most $k$, respectively.
Any undefined notation follows that of Bondy and Murty \cite{Bondy}.

A graph $G$ is $1$-immersed into a surface if it can be drawn on the surface so that each edge is crossed by at most one other edge. In particular, a graph is $1$-planar if it is $1$-immersed into the plane (i.e. has a plane $1$-immersion). The notion of $1$-planar-graph was introduced by Ringel \cite{Ringel} in the connection with problem of the simultaneous coloring of adjacent/incidence of vertices and faces of plane graphs. Ringel conjectured that each $1$-planar graph is $6$-vertex colorable, which was confirmed by Borodin \cite{Borodin}. Recently, Albertson and Mohar \cite{Albertson} investigated the list vertex coloring of graphs which can be $1$-immersed into a surface with positive genus. Borodin, et al. \cite{Borodin.acy} considered the acyclic vertex coloring of $1$-planar graphs and proved that each $1$-planar graph is acyclically $20$-vertex colorable. The structure of $1$-planar graphs was studied in \cite{Fabrici} by Fabrici and Madaras. They showed that the number of edges in a $1$-planar graph $G$ is bounded by $4v(G)-8$. This implies every $1$-planar graph contains a vertex of degree at most $7$. Furthermore, the bound $7$ is the best possible because of the existence of a $7$-regular $1$-planar graph (see Fig.1 in \cite{Fabrici}). In the same paper, they also derived the analogy of Kotzig theorem on light edges; it was proved that each $3$-connected $1$-planar graph $G$ contains an edge such that its endvertices are of degree at most $20$ in $G$; the bound $20$ is the best possible.

The aim of this paper is to exhibit a detailed structure of $1$-planar graphs which generalizes the result that every $1$-planar graph contains a vertex of degree at most $7$ in Section 2. By using this structure, we answer two questions on light graphs posed by Fabrici and Madaras \cite{Fabrici} in Section 3 and give a linear upper bound of acyclic edge chromatic number of $1$-planar graphs in
Section 4.

\section{Local structure of $1$-planar graphs}

To begin with, we introduce some basic definitions. Let $G$ be a $1$-planar graph. In the following, we always assume $G$ has been drawn on a plane so that every edge is crossed by at most one another edge and the number of crossings is as small as possible (such a dawning is called to be $proper$). So for each pair of edges $x_1y_1,x_2y_2$ that cross each other at a crossing point $z$, their end vertices are pairwise distinct. Let $\bar{C}(G)$ be the set of all crossing points and let $E_0(G)$ be the non-crossed edges in $G$. Then the $associated$ $plane$ $graph$ $\gx$ of $G$ is the plane graph such that $V(\gx)=V(G)\cup \bar{C}(G)$ and $E(\gx)=E_0(G)\cup \{xz,yz|xy\in E(G)\backslash E_0(G)~ {\rm{and~}} z~ {\rm{is~ the~ crossing~ point~ on~}} xy\}$. Thus the crossing points in $G$ become the real vertices in $\gx$ all having degree four. For convenience, we still call the new vertices in $\gx$ crossing vertices and use the notion $\bar{C}(\gx)$ to denote the set of crossing vertices in $\gx$.  A simple graph $G$ is $triangulated$ if every cycle of length $\geq 3$ has an edge joining two nonadjacent vertices of the cycle. We say $G_T$ is a $canonical$ $triangulation$ of a $1$-planar graph $G$ if $G_T$ is obtained from $G$ by the
following operations.\vspace{1mm}

\noindent  \textbf{Step 1}. For each pair of edges $ab,cd$ that cross each other at a point $s$, add edges $ac,cb,bd$ and $da$ "close to $s$", i.e. so that they form triangles $asc, csb, bsd$ and $dsa$ with empty interiors.

\noindent  \textbf{Step 2}. Delete all multiple edges.

\noindent  \textbf{Step 3}. If there are two edges that cross each other then delete one of them.

\noindent \textbf{Step 4}. Triangulate the planar graph obtained after the operation in Step 3 in any way.

\noindent \textbf{Step 5}. Add back the edges deleted in Step 3.

\vspace{1mm}\noindent Note that the associated planar graph $\gx_T$ of $G_T$ is a special triangulation of $\gx$ such that each crossing vertex remains to be of degree four. Also, each vertex $v$ in $\gx_T$ is incident with just $d_{\gx_T}(v)$ $3$-faces. Denote $v_1,\cdots,v_{d}$ to be the neighbors of $v$ in $\gx_T$ (in a cyclic order) and use the notations $v_i^+=v_{i+1}$, $v_i^-=v_{i-1}$, where $d=d_{\gx_T}(v)$ and $i$ is taken modulo $d$.

In the following, we use $c(v)$ to denote the number of crossing vertices which are adjacent to $v$ in $\gx_T$. Then we have the following observations. Since their proofs of them are trivial, we omit them here. In particular, the second observation uses the facts that $G_T$ admits no multiple edge and the drawing of $G_T$ minimizes the number of crossing. \vspace{1mm}

\begin{observation}\label{obs} For a canonical triangulation $G_T$ of a $1$-planar simple graph $G$, we have

{\rm(1)} Any two crossing vertices are not adjacent in $\gx_T$.

{\rm(2)} If $d_{G_T}(v)=3$, then $c(v)=0$.

{\rm(3)} If $d_{G_T}(v)=4$, then $c(v)\leq 1$.

{\rm(4)} If $d_{G_T}(v)\geq 5$, then $c(v)\leq \frac{d_{G_T}(v)}{2}$.

\end{observation}

Let $v\in V(G_T)$ and $u$ be a crossing vertex in $\gx_T$ such that $uv\in E(\gx_T)$. Then by the definitions of $u^+$ and $u^-$, we have $vu^+, vu^-\in E(\gx_T)$. Furthermore, the path $u^-uu^+$ in $\gx_T$ corresponds to the original edge $u^-u^+$ with a crossing point $u$ in $G_T$. Let $w$ be the neighbor of $v$ in $G_T$ so that $vw$ crosses $u^-u^+$ at $u$ in $G_T$. By the definition of $\gx_T$, we have $wu^-,wu^+\in E(\gx_T)$. We call $w$ the $mirror$-$neighbor$ of $v$ in $G_T$ and $u^-,u^+$ the $image$-$neighbors$ of $v$ in $G_T$. Other neighbors of $v$ in $G_T$ are called $normal$-$neighbors$. Sometimes when we say mirror vertex, image vertex and normal vertex, we refer to mirror neighbor, image neighbor and normal neighbor of $v$. The triangle $u^-wu^+$ in $G_T$ is called the $mirror$-$triangle$ incident with $v$. Since the neighbors of $v$ in $\gx_T$ can be listed in a cyclic order, via replacing the crossing vertex by the mirror vertex incident with it, the neighbors of $v$ in $G_T$ can be also listed in a cyclic order. Note that different crossing vertices are adjacent to different mirror vertices since multiple edges are forbidden in $G$. Let $v_1,\cdots,v_{d_{G_T}(v)}$ be the neighbors of $v$ in $G_T$ in a cyclic order. We define $\Omega(v_iv_j)=\{v_i,v_{i+1},\cdots,v_{j-1},v_{j}\}$, $v_i^{\triangleright}=v_{i+1}$ and $v_i^{\triangleleft}=v_{i-1}$, where $i$ is taken modulo $d_{G_T}(v)$. Note that $G_T$ is a canonical triangulation of $G$. Then $v_1v_2\cdots v_{d_{G_T}(v)}v_1$ is a cycle which is called the $associated$ $cycle$ of $v$, denoted by $C$. We call the path $P_i=v^i_1 v^i_2 \cdots v^i_{2t_i} v^i_{2t_i +1}$ a $segment$ of $C$ if (a) the elements of $\bigcup_{k=0}^{t_i}\{v^i_{2k+1}\}$ are image neighbors of $v$; (b) the elements of $\bigcup_{k=1}^{t_i}\{v^i_{2k}\}$ are mirror neighbors of $v$; (c) the triangles in $G_T$ of the form $v^i_{2k-1}v^i_{2k}v^i_{2k+1}$ where $1\leq k\leq t_i$ are mirror triangles incident with $v$ and (d) $v_1^{i\triangleleft}$ and $v_{2t_i +1}^{i\triangleright}$ are not mirror neighbors of $v$. Then $scope$ of a segment $P_i$ is defined to be the number of mirror triangles incident with $v$ using vertices in $V(P_i)$, denoted by $S(P_i)$. Then we easily have $S(P_i)=t_i$.

\begin{figure}
\begin{center}
  \includegraphics[width=16.0cm,height=5.5cm]{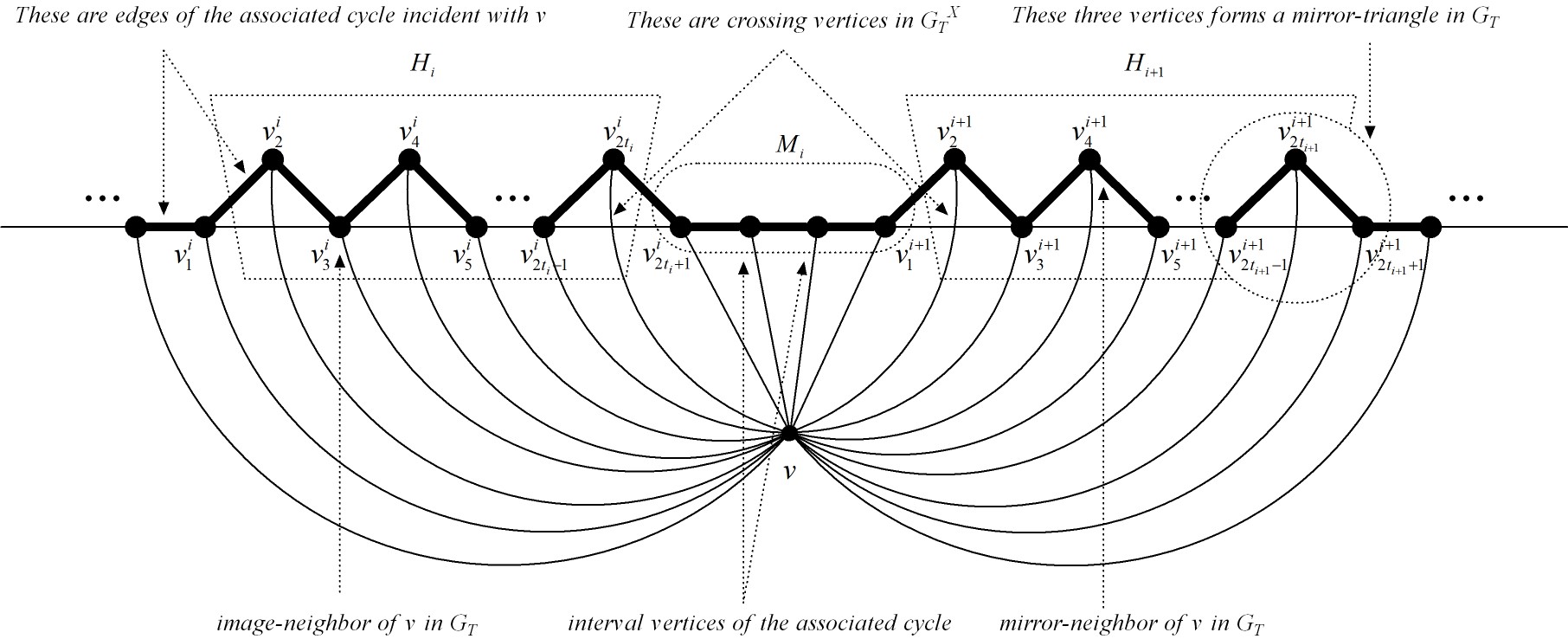}\\
  \caption{Some definitions on a canonical triangulation $G_T$ of a 1-planar graph $G$} \label{cal}
\end{center}
\end{figure}

Now we show the main result in this section.

\begin{theorem}\label{structure}
Let $G$ be a $1$-planar simple graph. Then there exists a vertex $v$ in $G$
with exactly $k$ neighbors $v_1,v_2,\cdots,v_k$ satisfying $d(v_1)\leq d(v_2)\leq
\cdots \leq d(v_k)$ such that one of the following is true.

{\rm(\C1)} $k\leq 2$,

{\rm(\C2)} $k=3$ with $d(v_1)\leq 35$,

{\rm(\C3)} $k=4$ with $d(v_1)\leq 19$ and $d(v_2)\leq 35$,

{\rm(\C4)} $k=5$ with $d(v_1)\leq 14$, $d(v_2)\leq 19$ and
$d(v_3)\leq 35$,

{\rm(\C5)} $k=6$ with $d(v_1)\leq 11$, $d(v_2)\leq 14$, $d(v_3)\leq
19$ and $d(v_4)\leq 35$,

{\rm(\C6)} $k=7$ with $d(v_1)\leq 8$, $d(v_2)\leq 11$, $d(v_3)\leq
14$, $d(v_4)\leq 19$ and $d(v_5)\leq 35$.
\end{theorem}

\begin{proof}
The theorem is proved by contradiction. Let $G$ be a simple $1$-planar graph with a fixed embedding in the plane, and suppose $G$ is a counterexample to the theorem. Note that if we add a new edge $e$ between two nonadjacent vertices in $G$ so that $G+e$ is still 1-planar graph, then $G+e$ shall also be a counterexample to the theorem. Hence in the following, without loss of generality, we always assume $G$ is 2-connected and $G=G_T$, where $G_T$ is a canonical triangulation of $G$ that has been draw on a plane properly. In other words, $G$ is just a canonical triangulation of itself. So in the next, there is no necessity to distinguish the two notions $G$ and $G_T$, and when we say $\gx$, we also refer to $\gx_T$. On the other hand, by the definition of associated plane graph, one can observe that $d_{\gx}(v)=d_G(v)$ when $v$ is not a crossing vertex. So in the detail proof below, we either need not to distinguish $d_{\gx}(v)$ and $d_G(v)$ when $v$ is a vertex of $G$, in which case we only use the notation $d(v)$ to represent both $d_{\gx}(v)$ and $d_G(v)$.

For a fixed vertex $v$ of $G$, we define $n_i(v)$ ($n_i$ in short) to be the number of neighbors of $v$ in $G$ which are of degree $i$. For a vertex set $S\subseteq V(G)$, let $n_i(S)=\sum_{v\in S}n_i(v)$. Denote $n_i^+(v)=\sum_{k=i}^{\Delta(G)}n_k(v)$ and $d=d(v)$. For a subgraph $H$ of $G$, $n_i(H)$ represents the number of $i$-vertices contained in $H$. Let $C$ be the associated cycle of $v$. Suppose $C$ has $n$ segments, denoted by $P_1,\cdots,P_n$ (in this cyclic order). Denote $t=\sum_{k=1}^n S(P_k)$. Let $M_i=\Omega(v^i_{2t_i+1}v^{i+1}_1)$ ($i=1,2,\cdots,n$). We call the vertex set $\bigcup_{k=1}^n \{M_k\backslash\{v^k_{2t_k+1},v^{k+1}_1\}\}$ the $interval$ of the associated cycle $C$. For each segment $P_i$, we define a graph $G_i$ so that $V(G_i)=V(P_i)$ and $E(G_i)=\bigcup_{k=1}^{t_i}\{v^i_{2k-1}v^i_{2k+1}\}\bigcup E(P_i)$. Let $H_i=G_i\backslash \{v^i_1,v^i_{2t_i+1}\}$. By the definition of
$P_i$, we have $H_i\subseteq G_i\subseteq G$. (see Fig. \ref{cal}).

A triangle $abc$ in $G$ is $light$ if $\max\{d(a),d(b),d(c)\}\leq 7$. Otherwise we say that $abc$ is $heavy$. Note that for the vertex $v$ described above, there are $t$ mirror triangles incident with it (recall the definition of the parameter $t$). Now, suppose $(t-x)$ of them are heavy and $x$ of them are light. We divide all the light mirror triangles incident with $v$ into three classes.

Class \Rmnum{1}. Triangles in the form $abc$ such that $a$ is mirror
vertex and $b,c$ are image vertices with $d(a)\leq 5$ and
$\min\{d(b),d(c)\}\geq 6$.

Class \Rmnum{2}. Triangles in the form $abc$ such that $a$ is mirror
vertex and $b,c$ are image vertices with $\min\{d(b),d(c)\}\leq 5$.

Class \Rmnum{3}. Triangles in the form $abc$ such that
$\min\{d(a),d(b),d(c)\}\geq 6$.

\noindent Denote the number of triangles belonging to Class \Rmnum{1}, \Rmnum{2} and \Rmnum{3} by $i,j$ and $x_2$, where
$i+j=x_1$ and $x_1+x_2=x$.

\vspace{1mm}\textbf{Claim 1}. $n_8^+\geq
\lceil\frac{t-x}{2}\rceil+\lceil\frac{j}{2}\rceil$.

Since each heavy mirror triangle incident with $v$ covers at least one $8^+$-vertex, there are at least $\lceil\frac{t-x}{2}\rceil$ $8^+$-vertices contained in heavy mirror triangles. And this lower bound is reachable only if each heavy mirror triangle covers exactly one $8^+$-vertex such that each pair of incident mirror triangles share one common $8^+$-vertex.

For each Class \Rmnum{2} light mirror triangle $abc$ such that $a$ is mirror vertex and $b,c$ are image vertices with $d(b)\leq 5$, since (\C1)-(\C4) are forbidden in $G$, we have $d(b)=5$ and another three neighbors of $b$ are all $36^+$-vertices. Let $C$ be the associated cycle of $v$. Then $1\leq|N_C(b)-\{a,c\}|\leq 2$. If $|N_C(b)-\{a,c\}|=1$, let $p\in N_C(b)-\{a,c\}$. Then $p$ is a normal vertex. So $p$ could be incident with at most two image vertices of degree no more than five. If $|N_C(b)-\{a,c\}|=2$, let $N_C(b)-\{a,c\}=\{p,q\}$. Then $bpq$ is a heavy mirror triangle with two $8^+$-vertices. In either case, we would account at least $\lceil\frac{j}{2}\rceil$ new $8^+$-vertices which are not counted
in the above step.

Hence, we have $n_8^+\geq \lceil\frac{t-x}{2}\rceil+\lceil\frac{j}{2}\rceil$.

\vspace{1mm}\textbf{Claim 2}. There is an integer $\mu\geq 0$ such
that $n_6+n_7= x+x_2-j+\mu$.

Since each Class \Rmnum{1} light mirror triangle contains two vertices either of degree $6$ or of degree $7$ and each Class \Rmnum{3} light mirror triangle contains three vertices either of degree $6$ or of degree $7$, we deduce that $n_6+n_7\geq
i+2x_2=x+x_2-j$.

\vspace{1mm}\textbf{Claim 3}. $n_5\leq
d-\frac{t+x-j}{2}-x_2-\mu-n_3-n_4$.

Since (\C1) is forbidden, we have $\delta(G)\geq 3$ and $n_5=d-n_3-n_4-n_6-n_7-n_8^+$. Note that $i+j+x_2=x$. We deduce from Claim 1 and Claim 2 that $n_5\leq
d-\frac{t+x-j}{2}-x_2-\mu-n_3-n_4$.

\vspace{1mm}\textbf{Claim 4}. $2n_3+n_4+t+x\leq d$.

Recall the definitions of $G_i,H_i$ and $M_i$ where $i=1,2,\cdots,n$. Each $G_i$ contains $t_i$ mirror triangles incident with $v$. Suppose $\alpha_i$ of them are light and $\beta_i$ of them are heavy. Then $\alpha_i+\beta_i=t_i$. Since (\C3) is forbidden, no light mirror triangle contains $4$-vertex. So the neighbors of $v$ in $G$ with degree $4$ are all contained either in the heavy mirror triangles or in the interval of $C$. Note that all the image vertices contained in $H_i$ are of degree at least five (see Figure \ref{cal}), so if there is a 4-vertex in $H_i$, then it must be a mirror vertex. In view of this, one can easily claim that $H_i$ contains at most $\beta_i$ $4$-vertices. Furthermore, if $|M_i|=2$, $d(v^{i+1}_1)=4$ and $d(v^i_{2t_i+1})=4$, then $v^i_{2t_i},v^i_{2t_{i}-1},v^{i+1}_2,v^{i+1}_3$ are all $36^+$-vertices since (\C3) is forbidden in $G$. So the triangles $v^i_{2t_{i}-1}v^i_{2t_i}v^i_{2t_i+1}$ and $v^{i+1}_1v^{i+1}_2v^{i+1}_3$ are both heavy. Then one can similarly claim that $H_j$ contains at most $\beta_j-1$ $4$-vertices, where
$j=i,i+1$.

By (2) of Observation \ref{obs} and the definition of $G_i$, if there are $3$-vertices on $C$, they must be on the interval. Suppose there are $\gamma_i$ $3$-vertices in $M_i$ where $\gamma_i\geq 0$. If $|M_i|\geq 3$, then $M_i$ contains at least $2\gamma_i+1$ non-$4$-vertices since (\C2) and (\C3) are forbidden. Here, note that neither $v^i_{2t_i+1}$ nor $v^{i+1}_1$ can be $3$-vertex by (2) of Observation \ref{obs} since each image vertex is adjacent to a crossing vertex in $\gx$. So $n_4(M_i)\leq |M_i|-(2\gamma_i+1)$ if $|M_i|\geq 3$. If $|M_i|=2$, then $\gamma_i=0$. So the above inequation on $n_4(M_i)$ holds unless $d(v^{i+1}_1)=4$ and $d(v^i_{2t_i+1})=4$. In this special case, this inequation becomes $n_4(M_i)\leq |M_i|-(2\gamma_i+1)+1$ indeed, but on the other hand, we have $n_4(H_i)\leq \beta_i-1$ and $n_4(H_{i+1})\leq \beta_{i+1}-1$ by the former arguments. Note that $N_G(v)=V(C)=\bigcup_{i=1}^{n} (V(H_i)\cup V(M_i))$ and $|\bigcup_{k=1}^n M_k|=d-\sum_{k=1}^n(2t_k-1)=d+n-2\sum_{k=1}^n(\alpha_k+\beta_k)=d+n-2x-2(t-x)=d+n-2t$. So we can deduce that $n_4=\sum_{k=1}^n n_4(H_k)+\sum_{k=1}^n n_4(M_k)\leq \sum_{k=1}^n \beta_k+\sum_{k=1}^n |M_k|-\sum_{k=1}^n (2\gamma_k+1)=(t-x)+(d+n-2t)-(2n_3+n)=d-x-t-2n_3$. Hence, we have
$2n_3+n_4+t+x\leq d$.\vspace{1mm}

Now we apply the discharging method to the associated planar graph $\gx$ of $G$. Since $\gx$ is a planar graph and $\sum_{v\in V(\gx)}(d_{\gx}(v)-4)=\sum_{v\in V(G)}(d_{\gx}(v)-4)+\sum_{v\in \bar{C}(\gx)}(d_{\gx}(v)-4)=\sum_{v\in V(G)}(d_{G}(v)-4)$. By Euler's formula, we have
$$\sum_{v\in V(G)}(d_G(v)-4)+\sum_{f\in F(\gx)}(d_{\gx}(f)-4)=-8.$$

Now we define $ch(x)$ to be the initial charge of $x\in V(G)\cup F(\gx)$. Let $ch(v)=d_G(v)-4$ for each vertex $v\in V(G)$ and let $ch(f)=d_{\gx}(f)-4$ and for each face $f\in F(\gx)$. It follows that $\sum_{x\in V(G)\cup F(\gx)}ch(x)=-8$. We now redistribute the initial charge $ch(x)$ and form a new charge $ch'(x)$ for each $x\in V(G)\cup F(\gx)$ by discharging method. Since our rules only move charge around, and do not affect the sum, we have $\sum_{x\in V(G)\cup F(\gx)}ch'(x)=\sum_{x\in V(G)\cup F(\gx)}ch(x)=-8$. A $3$-face in $\gx$ is $special$ if it is incident with one crossing vertex. Our discharging rules are defined as follows:\vspace{1mm}

\noindent(\R1) Each non-special $3$-face in $\gx$ receives $\frac{1}{3}$ from
each vertex incident with it;

\noindent(\R2) Each special $3$-face in $\gx$ receive $\frac{1}{2}$ from each
non-crossing vertex incident with it;

\noindent(\R3) Each vertex $v$ in $G$ with $9\leq d(v)\leq 11$ sends
$\frac{1}{21}$ to each adjacent $7$-vertex in $G$;

\noindent(\R4) Each vertex $v$ in $G$ with $12\leq d(v)\leq 14$ sends
$\frac{1}{18}$ to each adjacent $7$-vertex and $\frac{1}{6}$ to each
adjacent $6$-vertex in $G$;

\noindent(\R5) Each vertex $v$ in $G$ with $15\leq d(v)\leq 19$ sends
$\frac{1}{15}$ to each adjacent $7$-vertex, $\frac{1}{5}$ to each
adjacent $6$-vertex and $\frac{4}{15}$ to each adjacent $5$-vertex
in $G$;

\noindent(\R6) Each vertex $v$ in $G$ with $20\leq d(v)\leq 35$ sends
$\frac{1}{12}$ to each adjacent $7$-vertex, $\frac{1}{4}$ to each
adjacent $6$-vertex, $\frac{1}{3}$ to each adjacent $5$-vertex and
$\frac{5}{12}$ to each adjacent $4$-vertex in $G$;

\noindent(\R7) Each vertex $v$ in $G$ with $d(v)\geq 36$ sends
$\frac{1}{9}$ to each adjacent $7$-vertex, $\frac{1}{3}$ to each
adjacent $6$-vertex, $\frac{4}{9}$ to each adjacent $5$-vertex,
$\frac{5}{9}$ to each adjacent $4$-vertex and $\frac{2}{3}$ to each
adjacent $3$-vertex in $G$.

\vspace{1mm} Let $f$ be a face of $\gx$. Then $d_{\gx}(f)=3$. If $f$ is non-special, then by (\R1), $ch'(f)=ch(f)+3\times\frac{1}{3}=0$. If $f$ is special, then by Observation \ref{obs}, $f$ is incident with two non-crossing vertices. By (\R2), we have
$ch'(f)=ch(f)+2\times\frac{1}{2}=0$.

Let $v$ be a vertex of $G$. Since (\C1) is forbidden, we have $d(v)\geq 3$. Suppose $v$ is a $d$-vertex and has $d$ neighbors $v_1,\cdots,v_d$ in $G$ where $d(v_1)\leq\cdots\leq d(v_d)$. In the
following, we show $ch'(v)\geq 0$ for each such a vertex.

Suppose $d=3$. Since (\C2) is forbidden, $v$ is adjacent three $36^+$-vertices. Note that $v$ is not incident with any special $3$-face by (2) of Observation \ref{obs}. By (\R1) and (\R7), we have $ch'(v)\geq ch(v)-3\times \frac{1}{3}+3\times
\frac{2}{3}=0$.

Suppose $d=4$. Since $c(v)\leq 1$ by (3) of Observation \ref{obs}, $v$ is incident with at most two special $3$-faces. If $d(v_1)\geq 20$, then by (\R1), (\R2) and (\R6), we have $ch'(v)\geq ch(v)-2\times\frac{1}{2}-2\times\frac{1}{3}+4\times\frac{5}{12}=0$. If $d(v_1)\leq 19$, then $d(v_2)\geq 36$ since (\C3) is forbidden. So by (\R7) we have $ch'(v)\geq
ch(v)-2\times\frac{1}{2}-2\times\frac{1}{3}+3\times\frac{5}{9}=0$.

Suppose $d=5$. If $d(v_1)\geq 15$, then by (\R1), (\R2), (\R5) and (4) of Observation \ref{obs}, we have $ch'(v)\geq ch(v)-4\times\frac{1}{2}-\frac{1}{3}+5\times\frac{4}{15}=0$. So we may assume $d(v_1)\leq 14$. If $d(v_2)\geq 20$, then by (\R1), (\R2) and (\R6), we have $ch'(v)\geq ch(v)-4\times\frac{1}{2}-\frac{1}{3}+4\times\frac{1}{3}=0$. So we may assume $d(v_2)\leq 19$. Then $d(v_3)\geq 36$ for otherwise (\C4) occurs. In this case, by (\R1), (\R2) and (\R7), we also have $ch'(v)\geq
ch(v)-4\times\frac{1}{2}-\frac{1}{3}+3\times\frac{4}{9}=0$.

Suppose $d=6$. If $d(v_1)\geq 12$, then by (\R1), (\R2) and (\R4), we have $ch'(v)\geq ch(v)-6\times\frac{1}{2}+6\times\frac{1}{6}=0$. So we may assume $d(v_1)\leq 11$. If $d(v_2)\geq 15$, then by (\R1), (\R2) and (\R5), we have $ch'(v)\geq ch(v)-6\times\frac{1}{2}+5\times\frac{1}{5}=0$. So we may assume $d(v_2)\leq 14$. If $d(v_3)\geq 20$, then by (\R1), (\R2) and (\R6), we have $ch'(v)\geq ch(v)-6\times\frac{1}{2}+4\times\frac{1}{4}=0$. So we may assume $d(v_3)\leq 19$. Then $d(v_4)\geq 36$ for otherwise (\C5) occurs. In this case, by (\R1), (\R2) and (\R7), we also have $ch'(v)\geq
ch(v)-6\times\frac{1}{2}+3\times\frac{1}{3}=0$.

Suppose $d=7$. If $d(v_1)\geq 9$, then by (\R1), (\R2), (\R3) and (4) of Observation \ref{obs}, we have $ch'(v)\geq ch(v)-6\times\frac{1}{2}-\frac{1}{3}+7\times\frac{1}{21}=0$. So we may assume $d(v_1)\leq 8$. If $d(v_2)\geq 12$, then by (\R1), (\R2) and (\R4), we have $ch'(v)\geq ch(v)-6\times\frac{1}{2}-\frac{1}{3}+6\times\frac{1}{18}=0$. So we may assume $d(v_2)\leq 11$. If $d(v_3)\geq 15$, then by (\R1), (\R2) and (\R5), we have $ch'(v)\geq ch(v)-6\times\frac{1}{2}-\frac{1}{3}+5\times\frac{1}{15}=0$. So we may assume $d(v_3)\leq 14$. If $d(v_4)\geq 20$, then by (\R1), (\R2) and (\R6), we have $ch'(v)\geq ch(v)-6\times\frac{1}{2}-\frac{1}{3}+4\times\frac{1}{12}=0$. So we may assume $d(v_4)\leq 19$. Then $d(v_5)\geq 36$ for otherwise (\C6) occurs. In this case, by (\R1), (\R2) and (\R7), we also have $ch'(v)\geq
ch(v)-6\times\frac{1}{2}-\frac{1}{3}+3\times\frac{1}{9}=0$.

Suppose $d=8$. Then by (\R1)-(\R8), we have $ch'(v)\geq
ch(v)-8\times\frac{1}{2}=0$.

Suppose $9\leq d\leq 11$. Recall that $t$ is the number of mirror triangles incident with $v$. So $c(v)=t$. Note that $t\leq \frac{d}{2}$ and $n_7\leq d$. By (\R1), (\R2) and (\R3), we have $ch'(v)\geq ch(v)-\frac{2t}{2}-\frac{d-2t}{3}-\frac{n_7}{21}\geq
\frac{19}{42}d-4>0$.

Suppose $12\leq d\leq 14$. Note that $t\leq \frac{d}{2}$ and $n_6+n_7\leq d$. By (\R1),(\R2) and (\R4), we have $ch'(v)\geq ch(v)-\frac{2t}{2}-\frac{d-2t}{3}-\frac{n_6}{6}-\frac{n_7}{18}\geq
\frac{1}{3}d-4\geq 0$.

Suppose $15\leq d\leq 19$. Note that $t\leq \frac{d}{2}$. By (\R1),(\R2), (\R5) and Claims 2, 3, we have $ch'(v)=ch(v)-\frac{2t}{2}-\frac{d-2t}{3}-\frac{4n_5}{15}-\frac{n_6}{5}-\frac{n_7}{15}\geq \frac{2}{3}d-4-\frac{1}{3}t-\frac{1}{5}(x+x_2-j+\mu)-\frac{4}{15}(d-\frac{t+x-j}{2}-x_2-\mu-n_3-n_4)= \frac{2}{5}d-4-\frac{1}{5}t-\frac{1}{15}x+\frac{1}{15}(x_2+j+\mu+4n_3+4n_4)\geq
\frac{2}{5}d-4-\frac{4}{15}t\geq \frac{4}{15}d-4\geq 0$.

Suppose $20\leq d\leq 35$. Note that $t\leq \frac{d}{2}$. By (\R1),(\R2), (\R6) and Claims 2, 3, 4, we have $ch'(v)=ch(v)-\frac{2t}{2}-\frac{d-2t}{3}-\frac{n_7}{12}-\frac{n_6}{4}-\frac{n_5}{3}-\frac{5n_4}{12} \geq \frac{2}{3}d-4-\frac{1}{3}t-\frac{1}{4}(x+x_2-j+\mu)-\frac{1}{3}(d-\frac{t+x-j}{2}-x_2-\mu-n_3-n_4)-\frac{5n_4}{12} \geq \frac{1}{3}d-4-\frac{1}{12}(2t+x+n_4)+\frac{1}{3}n_3\geq \frac{1}{3}d-4-\frac{1}{12}(d-2n_3+t)+\frac{1}{3}n_3\geq
\frac{1}{4}d-4-\frac{1}{12}t\geq \frac{5}{24}d-4>0$.

Suppose $d\geq 36$. Note that $t\leq \frac{d}{2}$. By (\R1),(\R2), (\R7) and Claims 2, 3, 4, we have $ch'(v)=ch(v)-\frac{2t}{2}-\frac{d-2t}{3}-\frac{n_7}{9}-\frac{n_6}{3}-\frac{4n_5}{9}-\frac{5n_4}{9}-\frac{2n_3}{3} \geq \frac{2}{3}d-4-\frac{1}{3}t-\frac{1}{3}(x+x_2-j+\mu)-\frac{4}{9}(d-\frac{t+x-j}{2}-x_2-\mu-n_3-n_4)-\frac{5n_4}{9}-\frac{2n_3}{3} =\frac{2}{9}d-4-\frac{1}{9}(2n_3+n_4+t+x)+\frac{1}{9}(x_2+j+\mu)\geq
\frac{1}{9}d-4\geq 0$.

By the above arguments, we have $\sum_{x\in V(G)\cup F(\gx)}ch'(x)\geq 0$, a contradiction. Hence we have proved the
theorem.
\end{proof}

\section{Light graphs in $1$-planar graphs of the bounded degree}

Let $\mathcal{H}$ be a family of graphs and let $H$ be a connected graph. Let $\phi(H,\mathcal{H})$ be the smallest integer with the property that each graph $G\in \mathcal{H}$ contains a subgraph $K\cong H$ such that $\max_{x\in V(K)}\{d_G(x)\}\leq \phi(H,\mathcal{H})$. If such an integer does not exist, we write $\phi(H,\mathcal{H})=+\infty$. We say that the graph $H$ is $light$ in the family $\mathcal{H}$ if $\phi(H,\mathcal{H})<+\infty$. By $\mathcal{L}(\mathcal{H})$, we denote the set of light graphs in the family $\mathcal{H}$.

In the next, $P_k$ denotes a path with $k$ vertices and $S_k$ denotes a star with maximum degree $k$. We use the notation $\mathcal{P}^1_{\delta}$ for the family of all $1$-planar graphs of minimum degree at least $\delta$. In \cite{Fabrici}, Fabrici and Madaras showed that $\{P_1,P_2\}\subseteq\mathcal{L}(\mathcal{P}^1_{4})\subseteq \{P_1,P_2,P_3\}$ and $\{P_1,P_2,P_3\}\subseteq\mathcal{L}(\mathcal{P}^1_{5})\subseteq \{P_1,P_2,P_3,P_4,S_3\}$ and posed a few of open problems. Two
of them are stated as follows.\vspace{1mm}

Is $P_3\in \mathcal{L}(\mathcal{P}^1_{4})$ true?

Is $P_4,S_3\in \mathcal{L}(\mathcal{P}^1_{5})$ true?

\vspace{1mm}\noindent In this section, we partially answer these two
questions by applying the results in Section 2.

\begin{theorem} Let $G$ be a simple $1$-planar graph with minimum
degree $\delta\geq 4$. Then $G$ contains a $3$-path with all vertices of degree at most $35$.
\end{theorem}

\begin{proof}
By Theorem \ref{structure}, $G$ contains one of the configuration in
\{(\C3),(\C4),(\C5),(\C6)\} described in Section $2$. In each case,
we will find a path $uvw$ in $G$ such that
$\max\{d(u),d(v),d(w)\}\leq 35$.
\end{proof}

Similarly we can prove an analogous theorem.

\begin{theorem} Let $G$ be a simple $1$-planar graph with minimum
degree $\delta\geq5$. Then $G$ contains a $3$-star
with all vertices of degree at most $35$.
\end{theorem}

Hence we have the following many corollaries.

\begin{corollary}
$P_3\in \mathcal{L}(\mathcal{P}^1_{4})$.
\end{corollary}

\begin{corollary}
$\mathcal{L}(\mathcal{P}^1_{4})=\{P_1,P_2,P_3\}$.
\end{corollary}

\begin{corollary}
$S_3\in \mathcal{L}(\mathcal{P}^1_{5})$.
\end{corollary}

\section{Acyclic edge coloring of $1$-planar graphs}

A mapping $c$ from $E(G)$ to the sets of colors $\{1,\cdots,k\}$ is called a $proper$ $k$-$edge$ $coloring$ of $G$ provided any two adjacent edges receive different colors. The $edge$ $chromatic$ $number$ $\chi'(G)$ is the minimum number of colors needed to color the edges of G properly. A $proper$ $k$-$edge$ $coloring$ $c$ of $G$ is called an $acyclic$ $k$-$edge$-$coloring$ of $G$ if there are no bichromatic cycles in $G$ under the coloring $c$. The smallest number of colors such that $G$ has an acyclic edge coloring is called the $acyclic$ $chromatic$ $number$ of $G$, denoted by $\chi'_a(G)$. Acyclic edge coloring was introduced by Alon et al. \cite{Alon}, and they presented a linear upper bound on $\chi'_a(G)$. It was proved that $\chi'_a(G)\leq 64\Delta(G)$ holds for every graph, which was later improved to $16\Delta(G)$ by Molloy and Reed \cite{Molloy}. For planar graph $G$, A. Fiedorowicz et al. \cite{Fiedorowicz} proved that $\chi'_a(G)\leq 2\Delta(G)+29$. Recently, Hou et al. \cite{Hou} gave a better upon bound. They showed that $\chi'_a(G)\leq \max\{2\Delta(G)-2,\Delta(G)+22\}$ holds for each planar graph. Let $\phi$ be an edge coloring of $G$. For any vertex $v\in V(G)$, we define $\phi(v)=\{\phi(uv)|u\in N(v)\}$. In this section, we consider the acyclic edge coloring of $1$-planar graphs.

\begin{theorem} \label{acyclic}
Let $G$ be a $1$-planar simple graph. Then $\chi'_a(G)\leq
\max\{2\Delta(G)-2,\Delta(G)+83\}$.
\end{theorem}

\proof The theorem is proved by contradiction. Let $L$ stand for
$\max\{2\Delta(G)-2,\Delta(G)+83\}$. Suppose $G$ is a minimum
counterexample to the theorem. Then $G$ is $2$-connected and then
$\delta(G)\geq 2$.

\noindent\textbf{Case $1$}. $\delta(G)= 2$

Let $d(v)=2$ and $N(v)=\{v_1,v_2\}$. Suppose $v_1v_2\not\in E(G)$. By the minimality of $G$, the graph $G'=(G\backslash v)\cup \{v_1v_2\}$ has an acyclic $L$-edge coloring $\phi$ with color set $C$. Let $\tau(vv_1)=\phi(v_1v_2)$ and $\tau(vv_2)\in C\backslash \{\phi(v_2)\cup \tau(vv_1)\}$. For the edge $e\in E(G')-\{v_1v_2\}$, we remain $\tau(e)=\phi(e)$. Note that $|C\backslash \{\phi(v_2)\cup \tau(vv_1)\}|>0$. Then $\tau$ is an acyclic $L$-edge coloring of $G$, a contradiction. So $v_1v_2\in E(G)$. Let $G'=G\backslash v$. Then $G'$ has an acyclic $L$-edge coloring $\phi$ with color set $C$. Now we let $\tau(vv_1)\in S_1=C\backslash \{\phi(v_1)\cup\phi(v_2)\}$. Since $|C|\geq 2\Delta(G)-2$ and $|\phi(v_1)\cup\phi(v_2)|\leq 2\Delta(G)-3$, $|S_1|\geq 1$. Now we color $vv_2$ by $\tau(vv_2)\in S_2=C\backslash \{\phi(v_2)\cup \tau(vv_1)\}$. It is easy to see that $|S_2|>0$. For the edge $e\in E(G')$, we also remain $\tau(e)=\phi(e)$. Then $\tau$ is again an
acyclic $L$-edge coloring of $G$, a contradiction.

\noindent\textbf{Case $2$}. $\delta(G)\geq 3$

In this case, $G$ has one of the five configurations \{(\C2),(\C3),(\C4),(\C5),(\C6)\} which are described in Theorem \ref{structure}. Let $c_1=8$, $c_2=11$, $c_3=14$, $c_4=19$ and $c_5=35$. Suppose $G$ contains the $(d-1)$-th configuration, where
$d\in\{3,4,5,6,7\}$.

If $v_{d-1}v_d\not\in E(G)$, let $G'=(G\backslash v)\cup \{v_{d-1}v_d\}$. Otherwise, let $G'=G\backslash v$. Then $G'$ has an acyclic $L$-edge coloring $\phi$ with color set $C$. If $v_{d-1}v_d\not\in E(G)$, let $\tau(vv_{d-1})=\phi(v_{d-1}v_d)$. Otherwise, let $S_{d-1}=\phi(v_{d-1})\cup \phi(v_d)$. Now we color $vv_{d-1}$ by a color $\tau(vv_{d-1})\in C\backslash S_{d-1}$. Note that $|C|\geq 2\Delta(G)-2$ and $S_{d-1}\leq 2\Delta(G)-3$, we have $|C\backslash S_{d-1}|>0$. Let $S_d=\phi(v_1)\cup\cdots\cup\phi(v_{d-2})\cup \phi(v_d)$ and $S_i=\bigcup_{k=i}^{d-1}\phi(v_i)$ where $1\leq i\leq d-2$. Then we color $vv_d,vv_1,vv_2,\cdots, vv_{d-2}$ in turn as follows. Let $\tau(vv_d)\in T_d=C\backslash \{S_d\cup \tau(vv_{d-1})\}$. If $\tau(vv_{d-1})\not\in \phi(v_1)$, let $\tau(vv_1)\in T'_1=C\backslash \{(S_1\backslash \tau(vv_{d-1}))\cup \{\tau(vv_{d-1}),\tau(vv_d)\}\}$. Otherwise we let $\tau(vv_1)\in T_1=C\backslash \{S_1\cup \{\tau(vv_{d-1}),\tau(vv_d)\}\}$. At last, for each $2\leq i\leq d-2$, we let $\tau(vv_i)\in T_i=C\backslash \{S_i\cup \{\tau(vv_1),\cdots, \tau(vv_{i-1}),\tau(vv_{d-1}),\tau(vv_d)\}\}$. For the edge $e\in E(G')$, we still remain $\tau(e)=\phi(e)$. Note that $|T_{d-2}|\geq\cdots\geq |T_1|$, $|T'_1|\geq |T_1|$ and $\min\{|T_1|,|T_d|\}\geq L-(\sum_{k=1}^{d-2}(c_i-1)+\Delta(G))>0$. So this coloring $\tau$ does exist. It is easy to check that $\tau$ is proper and acyclic. So we have constructed a new coloring $\tau$ which is an acyclic $L$-edge coloring of $G$, a contradiction. This completes the proof of Theorem \ref{acyclic}. \hfill$\square$

\vspace{2mm}\noindent \textbf{Remark}. The proof of Theorem \ref{acyclic} above does not use recolorings, therefore, it actually yields a more general result as follows.
\begin{theorem}
Every simple $1$-planar graph $G$ is acyclically edge $L$-choosable where $L=\max\{2\Delta(G)-2,\Delta(G)+83\}$.
\end{theorem}


\begin{thebibliography}{10} \setlength{\itemsep}{0pt}

\bibitem{Albertson} M. O. Albertson, B. Mohar, Coloring vertices and faces of locally planar graphs, {\it Graphs and Combinatorics}, 22, (2006), 289--295.

\bibitem{Alon} N. Alon, C. J. H. McDiarmid, B. A. Reed, Acyclic coloring of graphs, {\it Random Structures and Algorithms}, 2, (1991), 277--288.

\bibitem{Bondy} J. A. Bondy and U. S. R. Murty, {\it Graph Theory with Applications}, North-Holland, New York, 1976.

\bibitem{Borodin} O. V. Borodin, Solution of Ringel's problems on the vertex-face coloring of plane graphs and on the coloring of $1$-planar graphs, {\it Diskret. Analiz}, 41(1984), 12--26 (in Russian).

\bibitem{Borodin.acy} O. V. Borodin, A. V. Kostochka, A. Raspaud, E. Sopena, Acyclic colouring of 1-planar graphs, {\it Discrete Applied Mathematics}, 114, (2001), 29--41.

\bibitem{Fabrici} I. Fabrici, T. Madaras, The structure of 1-planar graphs, {\it Discrete Mathematics}, 307, (2007), 854--865.

\bibitem{Fiedorowicz} A. Fiedorowicz, M. Halszczak, N. Narayanan, About acyclic edge colourings of planar graphs, {\it Information Processing Letters}, 108, (2008), 412--417.

\bibitem{Hou} J. Hou, J. L. Wu, G. Liu, B. Liu, Acyclic edge colorings of planar graphs and series-parallel graphs, {\it Science in China Series A: Mathematics}, 51, (2009), 605--616.

\bibitem{Molloy} M. Molloy, B. Reed, Further algorithmic aspects of the local lemma, {\it in: Proceedings of the 30th Annual ACM Symposium in Theory of Computing}, (1998), 524--529.

\bibitem{Ringel} G. Ringel, Ein sechsfarbenproblem auf der Kugel, {\it Abh. Math. Sem. Univ, Hamburg}, 29, (1965), 107--117.

\end{thebibliography}
\end{document}